\documentclass[unsortedaddress,nofootinbib,reprint]{revtex4-1}  

\usepackage[utf8]{inputenc}
\usepackage{amssymb,latexsym}
\usepackage{amsbsy} 
\usepackage{amsmath,amsthm}

\usepackage[shortlabels]{enumitem}
\usepackage{hyperref}

\theoremstyle{plain}
\newtheorem{theorem}{Theorem}

\newtheorem{lemma}[theorem]{Lemma}
\newtheorem{corollary}[theorem]{Corollary}

\theoremstyle{remark}
\newtheorem{remark}{Remark}

\theoremstyle{definition}
\newtheorem{definition}[theorem]{Definition}

\newcommand{\R}{{\mathbb R}}
\newcommand{\eps}{\varepsilon}

\newcommand \argmin {\mathop\mathrm{arg\,min}}
\newcommand \dom {\mathop\mathrm{dom}}
\newcommand \prox {\mathop\mathrm{prox}}
\newcommand \tr {\mathop\mathrm{tr}}

\newcommand \xks {x_\mathrm{KS}}

\begin{document}

\title{Generalized Kohn--Sham iteration on Banach spaces}

\author{Andre Laestadius}

\email{andre.laestadius@kjemi.uio.no}
\affiliation{
	Hylleraas Centre for Quantum Molecular Sciences, Department of Chemistry, University of Oslo, P.O.~Box 1033 Blindern, N-0315 Oslo, Norway}

\author{Markus Penz}

\affiliation{
Max Planck Institute for the Structure and Dynamics of Matter, Luruper Chausse 149, 22761 Hamburg, Germany}

\author{Erik I. Tellgren}

\affiliation{
Hylleraas Centre for Quantum Molecular Sciences, Department of Chemistry, University of Oslo, P.O.~Box 1033 Blindern, N-0315 Oslo, Norway}

\author{Michael Ruggenthaler}
\affiliation{
Max Planck Institute for the Structure and Dynamics of Matter, 
Luruper Chausse 149, 22761 Hamburg, Germany}

\author{Simen Kvaal}

\affiliation{
	Hylleraas Centre for Quantum Molecular Sciences, Department of Chemistry, University of Oslo, P.O.~Box 1033 Blindern, N-0315 Oslo, Norway}

\author{Trygve Helgaker}

\affiliation{
Hylleraas Centre for Quantum Molecular Sciences, Department of Chemistry, University of Oslo, P.O.~Box 1033 Blindern, N-0315 Oslo, Norway}
\affiliation{Centre for Advanced Study at the Norwegian Academy of Science and Letters,
Drammensveien 78, N-0271 Oslo, Norway
}

\date{\today}

\begin{abstract}
A detailed account of the Kohn--Sham algorithm from quantum chemistry,
formulated rigorously in the very general setting of convex analysis on Banach
spaces, is given here. Starting from a Levy--Lieb-type functional, its convex and
lower semi-continuous extension is regularized to obtain
differentiability. This extra layer allows to rigorously introduce, in contrast to the
common unregularized approach, a well-defined Kohn--Sham iteration scheme. 
Convergence in a weak sense is
then proven. This generalized formulation is applicable to a wide range of
different density-functional theories and possibly even to models
outside of quantum mechanics.
\end{abstract}

\maketitle

\section{Introduction}
Density-functional theory (DFT) is usually presented based on the Hohenberg--Kohn theorem~\cite{HK64} and received a thorough mathematical investigation in terms of convex analysis by Lieb \cite{Lieb83}. Yet it is the Kohn--Sham (KS) iteration scheme \cite{KS}, built upon the consequences of the Hohenberg--Kohn theorem, that makes it applicable to problems of quantum chemistry and where it developed an unprecedented utility. In Kohn--Sham theory a complex, interacting system with basic state variable $x$ (density) is compared against a simpler, non-interacting reference system that reproduces the exact same state $x$.
To achieve that, the reference system is complemented with an auxiliary potential $\xks^*$ that acts as a dual system variable and gets fixed by the connection of both systems via their respective energy functionals. 
Since no explicit expression for the Kohn--Sham potential $\xks^*$ is
at hand, a clever self-consistent iteration scheme is set up (see Section~\ref{sec:ks}).\\
\indent Nonetheless, up until now very few results about the convergence of this procedure are at hand. Among mathematical analyses of the self-consistent field procedure used to solve the Hartree--Fock or Kohn--Sham equations, we want to highlight in particular the work on the optimal damping algorithm (ODA) formulated in terms of one-particle reduced density matrices and Kohn--Sham matrices~\cite{Cances2000,Cances2000b,Cances2001}.\\
\indent In a setting where computationally efficient, approximate density
functionals are used, the Kohn--Sham wave function as the underlying
physical structure is crucial for providing a good first approximation
to the kinetic energy, since this is, in practice, very hard to
approximate from the density alone. 
However, in the rigorous mathematical setting of this work, where
properties of the exact density functionals are analyzed and employed,
the Kohn--Sham wave function itself does not show up. 
Somewhat unconventionally, one
can therefore choose to work with the exact energy of the Kohn--Sham
system as a (orbital-free) density functional. This makes the whole
analysis just as much applicable to orbital-free methods like
Thomas--Fermi theory. 
The analogue of the self-consistent field procedure in this setting is then a procedure that determines a sequence of densities and Kohn--Sham potentials that converges to a self-consistent pair. Working in this setting, Wagner et al.~\cite{Wagner2013} adapted the ODA and claimed,  unfortunately mistakenly, to show convergence to a self-consistent pair of a ground-state density and a Kohn--Sham potential.
Furthermore, the usual setting of Kohn--Sham theory is traditionally built on ill-defined quantities (cf. Section~\ref{sec:generalproblem}). This situation changed only recently when work by Kvaal et al.~\cite{Kvaal2014} showed how Moreau--Yosida regularization of convex functionals can be employed to make the Kohn--Sham theory rigorous. Specifically, Kvaal et al.~\cite{Kvaal2014} presented a well-defined iteration scheme based on a Hilbert space setting for density and potential variables.
However, the important problem of convergence was not addressed in this mathematically strict formulation of the theory.\\
\indent In this work we aim at closing gaps of Kohn--Sham theory while at the same time considerably extending its scope.
Important results include the following:
(i)~In Section \ref{sec:Mono}, Lemma \ref{lemma:monotone} gives an abstract, convex analytical version of the Hohenberg--Kohn theorem as an automatic consequence of strict monotonicity of the superdifferential of the energy functional.
(ii)~Section \ref{sec:my} is a full generalization of the regularization tools of Kvaal et al.~\cite{Kvaal2014} to reflexive Banach spaces that will be employed to formulate a rigorous Kohn--Sham iteration scheme in Section~\ref{sec:ks}. We especially discuss G\^{a}teaux and Fr\'{e}chet differentiability of regularized functionals (Theorem~\ref{thm:258cropped}) and the connection between solutions to the regularized ground-state problem and ``physical'' solutions (Corollary~\ref{cor:proximal}).
(iii)~Finally, Section~\ref{sec:ks} with Theorem~\ref{th:ks} gives the Kohn--Sham iteration scheme and shows a weak form of convergence. 
This result is based on an idea by Wagner et al.~\cite{Wagner2013} but here repeated for the well-defined setting of regularized functionals.
Gaps in the previously attempted convergence proof in~\cite{Wagner2013} are uncovered in Remark~\ref{remark:wagner}.
A generalized Banach space formulation of DFT has the advantage of
laying possible solid foundations to many types of ground-state DFTs,
including more recent developments in current-density-functional
theory \cite{Tellgren2017}, quantum-electrodynamical DFT \cite{Rugg2017}, thermal DFT
\cite{Eich2016}, and reduced density matrix functional theory
\cite{Pernal2015,Giesbertz2017}. 
Note that, even though most of the presentation has a link to quantum mechanics, the formulation here is kept in a general fashion that might prove valuable for applications to other fields.
The usual DFT setting for many-electron quantum systems is presented in Section~\ref{sec:applications}, where we also discuss extensions of the standard theory that include magnetic fields.

\section{Preliminaries}
\subsection{The general problem}
\label{sec:generalproblem}

The principal state variable $x$ (density) is chosen from a real Banach space $X$.
The dual space of $X$ (consisting of potentials) is 
\[
X^* = \{ x^*: X \to \mathbb R \mid \text{$x^*$ linear and continuous}    \}
\]
and we write $\langle x^*,x\rangle = x^*(x)$ for the dual pairing between $x\in X$ and $x^*\in X^*$. 
In DFT, one is studying a variational problem that consists of finding the $N$-electron ground-state density corresponding
to minimal energy (see Section~\ref{sec:standard-dft}). 
For the general problem we consider an energy functional $E: X^*\to \mathbb R$ 
given by 
\begin{align}
\label{eq:Edef}
E(x^*) = \inf \{ \tilde F(x) + \langle x^*,x \rangle \mid x \in \tilde X   \},
\end{align}
for some functional $\tilde F: \tilde X \to \mathbb R$, $\tilde X \subseteq X$, that originates from the underlying physical problem. In such models $\tilde F(x)$ stands for all energy contributions of internal effects, while the dual pairing $\langle x^*,x \rangle$ represents the potential energy that is seen as an external and controllable effect. 
The domain of $\tilde F$ is limited to a certain set of ``physical'' densities $\tilde X$ that by themselves usually do not form a linear space and $\tilde F(x)$ represents the minimal possible internal energy for a given state $x$. Such a $\tilde F$, called the Levy--Lieb functional~\cite{Lieb83,Levy79} in the DFT literature, is in general not convex or lower semi-continuous.
Therefore one introduces (borrowing terminology from DFT) the universal Lieb functional, $F:X\to \mathbb R\cup \{+\infty\}$, as \cite{Lieb83}
\begin{align}
\label{eq:Fdef}
F(x) = \sup\{ E(x^*) - \langle x^*,x\rangle \mid x^*\in X^*   \},
\end{align}
in terms of which we may obtain the energy as
\begin{align}
\label{eq:EdefF}
E(x^*) = \inf \{  F(x) + \langle x^*,x \rangle \mid x \in X   \}.
\end{align}
The functional $F$ is by construction convex and lower semi-continuous \cite[Theorem~3.6]{Lieb83} and has $F(x) = +\infty$ whenever $x \in X \setminus \tilde X$ (corresponding to ``unphysical'' densities). 
If one chooses $x^*\in X^*$ such that a minimizer $x\in \tilde X$ for \eqref{eq:Edef} exists and thus $E(x^*) = \tilde F(x) + \langle x^*, x\rangle$, then it follows after insertion into \eqref{eq:Fdef} that $F(x) = \tilde F(x)$. 
Each minimizer in~\eqref{eq:Edef} is therefore a minimizer in~\eqref{eq:EdefF} but the converse does not hold. A minimizer in~\eqref{eq:EdefF} need not be a minimizer of~\eqref{eq:Edef} and in general 
$F \leq \tilde F$ on $\tilde X$.\\
\indent To find the state $x$ of minimal energy in \eqref{eq:Edef} by just relying on $F$, we must determine
\begin{equation}\label{eq:argminenergy}
\argmin \{  F(x) + \langle x^*,x\rangle \mid x\in X  \},
\end{equation}
calling for differentiation of $F(x) + \langle x^*,x\rangle$ with respect to $x$.
But even though $F$ already has some nice properties, one cannot assume that $F$ is differentiable \cite{Lammert2007}. Yet the more general notion of ``subdifferential'' (see Definition~\ref{def:diff}) gives a non-unique, implicit answer to \eqref{eq:argminenergy}, namely
\begin{equation}\label{eq:varprob1}
-x^* \in \underline \partial F(x) \iff E(x^*) =  F(x) + \langle x^*,x \rangle .
\end{equation}
For a given density $x$, a non-empty subdifferential means that this state is ``$v$-representable'' 
(we will use this term commonly found in the literature instead of
calling it ``$x^*$-representable''), meaning that
it can be found as the minimizer to \eqref{eq:Edef} with a corresponding potential in $X^*$. 
(Note that the concept of $v$-representability depends on the choice of $X$ and $X^*$.)
For given $x^*$, \eqref{eq:varprob1} is the variational problem that finally gets approximated with the aid of a Kohn--Sham reference system. To facilitate the iteration scheme with well-defined quantities, Moreau--Yosida regularization (see Section~\ref{sec:my}) is employed. This yields a functional $F_\eps$ that is also functionally differentiable.
But before that, we give the relation between the fundamental functionals $E$ and $F$ a strict mathematical meaning.

\subsection{Convex conjugates}
To start with, we say that a convex function $f : X\to \mathbb R\cup \{+\infty\}$ is proper if not identical to $+\infty$.
Let $\Gamma_0(X)$ denote the set of proper, convex, lower semi-continuous functions $X \to \mathbb R\cup \{+\infty\}$. 
On the dual side, $\Gamma_0^*(X^*)$ is the set of proper convex and weak-* lower semi-continuous functions $X^* \to \mathbb R\cup \{+\infty\}$.
We also introduce the sets $\Gamma(X) = \Gamma_0(X) \cup \{\pm \infty\}$ and $\Gamma^*(X^*) = \Gamma_0^*(X^*) \cup \{\pm \infty\}$.
Collectively, these functions are known as the closed convex functions on $X$. Likewise, $\Gamma_0^*(X^*)$ contains the closed convex functions on $X^*$.\\
\indent Key in the presented framework will be two close relatives to the convex conjugate (Legendre--Fenchel transform), a non-standard definition introduced in \cite{Lieb83,Kvaal2014}. 
For $f: X \to \mathbb R \cup \{\pm \infty \}$ and $g: X^* \to \mathbb R \cup \{\pm \infty \}$, we define the (skew) conjugate functions: 
%
\begin{align*}
f^\wedge(x^*) &= \inf \{ f(x) + \langle x^*,x\rangle \mid x\in X \}, \,\, f^\wedge \in -\Gamma^*(X^*), \\
g^\vee(x) &= \sup\{ g(x^*) -\langle x^*,x\rangle \mid x^* \in X^* \}, \,\, g^\vee \in \Gamma(X). 
\end{align*}
%
The following non-trivial result relates $f$ and $g$ to their biconjugates $(f^\wedge)^\vee$ and
$(g^\vee)^\wedge$, respectively:
\begin{theorem}[Fenchel--Moreau biconjugation]
\begin{align*}
f=(f^\wedge)^\vee &\iff f \in \Gamma(X), \\
g=(g^\vee)^\wedge & \iff g \in -\Gamma^*(X^*).
\end{align*}
\end{theorem}

This theorem sets up a one-to-one correspondence between
the closed convex functions on $X$ and the closed concave functions on $X^*$. 
If $f$ and $g$ are  not closed convex and concave, respectively,
then the weaker results $f \geq (f^\wedge)^\vee$ and $g \leq (g^\vee)^\wedge$ hold
and we may think of biconjugation as performing ``gamma regularizations'' of $f$ and $g$. 
\\
\indent We have already used the skew conjugations in \eqref{eq:Edef}--\eqref{eq:EdefF} above, which may now be more
succinctly written in the form (having extended $\tilde F$ to to the whole $X$ by setting it equal to $+\infty$ outside $\tilde X$):
\begin{equation*}
E = \tilde{F}^\wedge,\quad
F = E^\vee, \quad 
E = F^\wedge .
\end{equation*}
We note that the ground-state energy is closed concave
$E \in -\Gamma_0^*(X^*)$ and that the closed convex Lieb functional
$F = (\tilde F^\wedge)^\vee \in \Gamma_0(X)$
is the gamma regularization of the Levy--Lieb functional $\tilde F \notin \Gamma(X)$.
(In the DFT context, the Lieb functional is the gamma regularization of any 
admissible density functional in the Hohenberg--Kohn variational principle.)

\subsection{Banach space derivatives}

\begin{definition} \label{def:diff}
	We give the following notions for derivatives of a function $f:X \to \mathbb R \cup \{+\infty \}$:
	\begin{enumerate}[(i)]
		\item $f$ is Gâteaux differentiable at $x\in \dom(f)$ if there exists a $\nabla f(x)\in X^*$ such that for all $u\in X$
		\begin{align*}
		\lim_{t\to 0}  \frac 1 t (f(x + tu) - f(x))= \langle \nabla f(x),u\rangle.
		\end{align*}
		
		\item $f$ is Fréchet differentiable at $x\in \dom(f)$ if there exists a $\nabla f(x) \in X^*$ such that 
		\begin{align*}
		\lim_{h\to 0}  \Vert h \Vert^{-1} (f(x + h) - f(x) - \langle \nabla f(x), h\rangle)=0.
		\end{align*}
		
		\item $f$ proper convex is subdifferentiable at $x \in X$ if it has a nonempty subdifferential $\underline \partial f(x) \subset X^*$ at $x$ given by
		\begin{align*}
		\underline \partial f(x) = \{& x^* \in X^* \mid 
		\forall y\in X:\\ & f(x)-f(y) \leq \langle x^*,x-y\rangle \}.
		\end{align*}
The elements of $\underline{ \partial} f(x)$ are known as the subgradients of $f$ at $x$.
	\end{enumerate}
\end{definition}

We denote by $\dom(f)$ the effective domain  of $f$ (the subset of $X$ where $f$ is finite) and by $\dom(\underline \partial f)$
the domain of subdifferentiability of $f$ (the subset of $X$ where $\underline \partial f$ is non-empty). We note that $\dom(\underline \partial f) \subset \dom(f)$
but $\underline \partial f(x) = \emptyset$ may happen also for $x \in \dom(f)$. 
If $f$ is closed convex, then $\dom(\underline \partial f)$ is dense in $\dom(f)$  by the Br{\o}ndsted--Rockafellar theorem.\\
\indent On the dual space, we mostly work with proper concave functions, whose superdifferentials are defined by 
\begin{equation}\label{eq:superdiff}
\begin{aligned}
\overline {\partial} g(x^*) = \{& x \in X \mid 
\forall y^* \in X^*\!:\\& g(x^*)-g(y^*) \geq \langle x^*-y^*, x\rangle \}.
\end{aligned}
\end{equation}
Note that $\overline {\partial} g(x^*)$ is by definition a subset of $X$ rather than $X^{**} \supseteq X$. The spaces $X$ and $X^{**}$ only match if they are reflexive, which will be an important property later. In an obvious manner, we may also define supergradients on $X$ and subgradients on $X^*$.

\subsection{Optimality conditions}

Sub- and superdifferentiation are precisely the tools needed to characterize optimality of convex and concave conjugations. The following results are easily established from the definition of
the sub- and superdifferentials:
\begin{lemma}\label{lemma:optimality}
If $f: X \to \mathbb R \cup \{\pm \infty \}$ proper convex and $g: X^* \to \mathbb R \cup \{\pm \infty \}$ proper concave then 
\begin{align*}
-x^* \in \underline\partial f(x) &\iff f^\wedge(x^*) = f(x) + \langle x^*, x \rangle, \\
x \in {\overline \partial} g(x^*) &\iff g^\vee(x) = g(x^*) - \langle x^*, x \rangle.
\end{align*}
\end{lemma}
\begin{proof}From the definitions of subdifferential and conjugate function, we have
\begin{align*}
-&x^* \in \underline\partial f(x) \\ &\iff \forall y \in X: f(y) + \langle x^*, y \rangle \geq f(x) + \langle x^*,x \rangle \\
 & \iff \inf \{ f(y) + \langle x^*, y \rangle \mid y \in X \} \geq f(x) + \langle x^*,x \rangle \\
 & \iff \inf \{ f(y) + \langle x^*, y \rangle \mid y \in X \} = f(x) + \langle x^*,x \rangle  \\ 
 & \iff f^\wedge(x^*) = f(x) + \langle x^*,x \rangle .
\end{align*}
The result for $g$ follows in a similar manner.
\end{proof}
For closed convex functions, we obtain by setting $g=f^\wedge$ in Lemma~\ref{lemma:optimality} and applying the biconjugation theorem:
\begin{lemma}\label{lemma:equivdiff}
\cite[Proposition~2.33]{Babru-Precupanu}
	For $f \in \Gamma_0(X)$ we have  
	\begin{align*}
	-x^* \in \underline\partial f(x) \iff x \in \overline \partial f^\wedge(x^*).
	\end{align*}
\end{lemma}
Finding a solution to $-x^* \in \underline \partial f(x)$ is therefore equivalent to finding a solution to $x \in \overline \partial f^\wedge(x^*)$.
This switch between primal and dual problems will be used several times in this work.

%

\subsection{Monotonicity of subdifferentials}
\label{sec:Mono}

\begin{lemma}\label{lemma:monotone}
	Let $f:X \to \mathbb R \cup \{+\infty \}$ be proper convex. Then $\underline\partial f$ is a monotone operator, that is,
for any $x,y \in X$, 
each pair $x^* \in \underline\partial f(x)$ and  $y^* \in \underline\partial f(y)$ satisfies
	\[
	\langle x^*-y^*,x-y \rangle \geq 0. \label{eq:mono} 
	\]
	If $f$ is strictly convex, then $\underline\partial f$ is strictly monotone, meaning that the inequality holds strictly for $x \neq y$.
	For concave functions, the inequality is reversed.
\end{lemma}

\begin{proof}
Let $x,y \in X$ be subdifferentiable points and select subgradients $x^* \in \underline\partial f(x),y^* \in \underline\partial f(y)$. Then
\begin{align*}
f(x)-f(y) &\leq \langle x^*, x - y \rangle, \\
f(y)-f(x) &\leq \langle y^*, y - x \rangle = \langle -y^*, x - y \rangle. 
\end{align*}
Adding these subgradient inequalities together, we obtain monotonicity.
If $f$ is strictly convex, then the subgradient inequalities hold strictly for $x \neq y$, yielding
strict monotonicity.
\end{proof}

The above lemma is a rigorous convex analytical version of a more
physically motivated lemma in Wagner et al.~\cite{Wagner2013} and is equivalent to the Hohenberg--Kohn theorem. To see that, assume that $E$ is
strictly concave (which is not true generally but will be the case for
the regularized version, see Remark~\ref{remark:strictly-concave}) and take $x \in
\overline\partial E(x^*)$, $y \in \overline\partial E(y^*)$. Then by 
Lemma~\ref{lemma:monotone} we get strict monotonicity of $\overline\partial E$, and with the roles of $x,y$ and $x^*,y^*$
interchanged, it follows from $x^* \neq y^*$ that $x \neq y$. 
This means that different potentials always lead to different states and that the corresponding mapping 
$x^* \mapsto x \in \overline\partial E(x^*)$ is
injective and thus invertible on the restricted codomain of
$v$-representable states. This is just the main statement of the
Hohenberg--Kohn theorem (including possible degeneracies), which 
is discussed again in the standard setting of DFT in Section~\ref{sec:standard-dft}.

\subsection{Moreau--Yosida regularization}
\label{sec:my}

To perform regularization of functionals we have to demand additional properties for $X,X^*$. Most importantly, the spaces have to be reflexive.  Note that $X$ reflexive implies $X^*$ reflexive and vice versa.
Additional conditions will include strict and (in Theorem {\ref{thm:258cropped}}) uniform convexity.
In a strictly convex space the connecting line segment between surface points of the unit ball lies strictly inside the ball, while in a uniformly convex space the distance of the middle point of the line segment from the surface is not only non-zero but also depends only on the length of the segment, not on the chosen points themselves. Obviously, a uniformly convex space is also strictly convex. (For further details see \mbox{\cite[Sections~1.2.3 and 1.2.4]{Babru-Precupanu}}.)\\
\indent It is interesting to note that uniform convexity of $X$ implies reflexivity (Milman--Pettis theorem).
By a theorem due to Asplund,  
if $X$ is assumed reflexive, an equivalent norm can be chosen such that $X$ and $X^*$ are strictly convex (yet not uniformly convex). In the following, we still keep the additional assumption of strict convexity in order to maintain the respective given norms of the Banach spaces at hand. Important uniformly convex spaces are the Lebesgue spaces $L^p$ with $1 < p < \infty$, but not $L^1$ or $L^\infty$, which are not even reflexive or strictly convex. Spaces that are reflexive but admit no equivalent norm that makes them uniformly convex exist \mbox{\cite{Day1941}}, but do not occur ``naturally''.\\
\indent The following simple functional will be the centerpiece of the regularization performed here.

\begin{definition} \label{def:MYreg} Set $\phi(\cdot) = \frac 1 2 \Vert \cdot \Vert^2$ on $X$ and $X^*$. Then for $\eps>0$ 
	the Moreau--Yosida regularization of a $f \in \Gamma_0(X)$ is defined as
	\begin{align*}
	f_\eps(x) = \inf \{ f(y) + \eps^{-1}\phi( x-y) \mid y\in X \}.
	\end{align*}
\end{definition}
\begin{remark}\label{remark:strict}
	The function $\phi$ is strictly convex if and only if the respective Banach space is strictly convex \cite[Proposition~1.103]{Babru-Precupanu}.
\end{remark}
\begin{remark}\label{remark:infconvolution}
	Alternatively we can define $f_\eps =f\,\square\, \eps^{-1}\phi$, where the box notation stands for the \emph{infimal convolution} of $f$ and $g$ and is given by 
	\[
	(f\, \square \,g)(x) = \inf\{ f(y) + g(x-y) \mid y\in X \}.
	\]
\end{remark}
A minimal value of $f$ (if it exists) is preserved at the same location when passing over to $f_\eps$.
The infimum in the definition above is always uniquely attained, $f_\eps \in \Gamma_0(X)$ and is everywhere finite on $X$, features for which reflexivity and strict convexity of $X,X^*$ is imperative \cite[Section 2.2.3]{Babru-Precupanu}.
This unique minimizer for any given $x$ gives rise to the definition of the proximal mapping
\begin{align*}
\prox_{\eps f}(x) &= \argmin\{ f(y) + \eps^{-1}\phi(x-y) \mid y\in X \} \\
&= \argmin\{ \eps f(y) + \phi(x-y) \mid y\in X \}.
\end{align*}

\begin{definition} \label{def:dualitymap}
	The duality map on $X$ is
	\begin{equation}\label{eq:dualitymap}
	J(x) = \{ x^* \in X^* \mid \|x^*\| = \|x\|, \langle x^*,x \rangle = \|x\|^2 \}
	\end{equation}
	and assigns to each state a set of dual elements. 
\end{definition}

The duality map is a bijective mapping $X \to X^*$ in the case of $X,X^*$ reflexive and strictly convex \cite[Proposition~1.117]{Babru-Precupanu}.
Of particular interest is the subdifferential of $\phi$ that yields just the duality map \cite[Example~2.32]{Babru-Precupanu}
\begin{align}\label{eq:phiSub}
\underline\partial \phi(x) = J(x).
\end{align}
Under the assumption that $X$ is reflexive, we can define $J^{-1}$ on all of $X^*$ \cite[Proposition~1.117(iv)]{Babru-Precupanu} and then have \eqref{eq:phiSub} also in the reverse direction
\begin{equation} \label{eq:phiSubJ}
\underline\partial \phi(x^*) = J^{-1}(x^*).
\end{equation}

\begin{lemma} \label{lemma:Hilt}
	If $X^*$ is strictly convex then $\phi$ is Gâteaux differentiable, if $X^*$ is uniformly convex then $\phi$ is Fréchet differentiable. In both cases the derivative is the duality map $J$.\footnote{Note that in the published version of this article Lemma~\ref{lemma:Hilt} is given as Proposition 8 and without proof, instead with a reference to \cite{Megginson1998}.}
\end{lemma}

\begin{proof}
That the general subdifferential of $\phi$ is identical to the duality map was already mentioned in \eqref{eq:phiSub}. If $X^*$ is strictly convex then by \cite[Proposition 1.117~(iii)]{Babru-Precupanu} this subdifferential is single-valued which already establishes Gâteaux differentiability \cite[Proposition 2.40]{Babru-Precupanu}. If furthermore $X^*$ is uniformly convex then by \cite[Proposition 1.117~(vi)]{Babru-Precupanu} the derivative is also (uniformly) continuous which implies Fréchet differentiability \cite[Lemma 34.3]{Blanchard-Bruening}.
\end{proof}

If we recall that the Moreau--Yosida regularization of a functional corresponds to the infimal convolution with $\phi$ (Remark~\ref{remark:infconvolution}), we can expect that the regularity properties of $\phi$ that follow from Lemma~\ref{lemma:Hilt} are taken over to the regularized functional $f_\eps = f \Box \eps^{-1}\phi$. This means $f_\eps$ should be Gâteaux differentiable if $X^*$ is strictly convex and Fréchet differentiable if $X^*$ is uniformly convex, statements that get proven by the following theorem.

\begin{theorem}\label{thm:258cropped}
	Suppose $f\in \Gamma_0(X)$ and that $X,X^*$ are reflexive and strictly convex. Then $f_\eps$ is Gâteaux differentiable on $X$ and
	\[
	\nabla f_\eps(x)= \eps^{-1}J(x - \prox_{\eps f}(x)).
	\]
	If $X^*$ is uniformly convex, then $f_\eps$ is even Fréchet differentiable.
\end{theorem}

\begin{proof}
	The proof of the first part can be found in \cite{Babru-Precupanu}, Theorem 2.58, with the derivative given by (2.57) there. That the derivative can indeed be evaluated by application of the proximal mapping is shown in the proof of Theorem 2.58. Proposition 1.146 (i) in \cite{Babru-Precupanu} further states that if $X^*$ is uniformly convex then $\nabla f_\eps : X \rightarrow X^*$ is continuous, which implies Fréchet differentiability \cite[Lemma 34.3]{Blanchard-Bruening}.
\end{proof}

We now turn our attention to the convex conjugate pair $E= F^\wedge$ from \eqref{eq:Edef} and \eqref{eq:Fdef} and their regularized versions
\begin{align*}
F_\eps(x) &= \inf\{ F(y) + \eps^{-1}\phi(x-y) \mid y\in X \},\\
E_\eps(x^*) &= (F_\eps)^\wedge(x^*).
\end{align*}
Note especially that $E_\eps$ is not the regularization of $E$ but rather the transformation of the regularized $F_\eps$.

\begin{theorem} \label{thm:Ereg}
	Suppose $F\in \Gamma_0(X)$ and $X$ reflexive then
	\begin{align}
	E_\eps(x^*) &= E(x^*) - \eps \phi( x^* ) \quad\mbox{and} \label{eq:Ereg1} \\
	\overline\partial E_\eps(x^*) &= \overline\partial E(x^*) - \eps J^{-1}(x^*). \label{eq:Ereg2}
	\end{align}
\end{theorem}
\begin{proof}
	To prove \eqref{eq:Ereg1}, note that by definition
	\begin{align*}
	E_\eps(x^*) =& (F_\eps)^\wedge(x^*) \\
	=& \inf\{ F(y) + \eps^{-1}\phi(x-y) + \langle x^*,x \rangle \mid x,y \in X \} \\
	=& \inf\{ F(y) + \langle x^*,y \rangle \\
	&+ \eps^{-1}\phi(x-y) + \langle x^*,x-y \rangle \mid x,y \in X \} \\
	=& \inf\{ F(y) + \langle x^*,y \rangle \mid y \in X\} \\
	&+ \inf\{ \eps^{-1}\phi(z) + \langle x^*,z \rangle \mid z \in X \} \\
	=& F^\wedge(x^*) + (\eps^{-1}\phi)^\wedge(x^*).
	\end{align*}
	From $\phi(x) =  \Vert x \Vert^2/2$ it follows that  $\phi^\wedge(x^*) = -\Vert x^* \Vert^2/2$. Furthermore, the scaling relation $(\lambda f)^\wedge(x^*) = \lambda f^\wedge(x^*/\lambda)$, $\lambda>0$, gives $(\eps^{-1}\phi)^\wedge(x^*)=-\eps \Vert x^* \Vert^2/2$. Thus
	\[
	(\eps^{-1}\phi)^\wedge(x^*) = -\eps \phi(x^*)  
	\]
	and we can conclude that \eqref{eq:Ereg1} holds.
%
  Eq.~\eqref{eq:Ereg2} follows directly from forming the superdifferential of \eqref{eq:Ereg1} and inserting \eqref{eq:phiSubJ}.
\end{proof}

\begin{remark}\label{remark:strictly-concave}
	Eq.~\eqref{eq:Ereg1} shows that $E_\eps$ is \emph{strictly} concave as the sum of a concave function and the strictly concave $-\eps \phi$ (if $X^*$ is strictly convex). From strict concavity follows that $E_\eps(x^*) - \langle x^*,x \rangle$ attains its maximum at one point only, so the regularized version of \eqref{eq:Fdef} has a unique maximizer. 
        The subdifferential of the conjugate $F_\eps = (E_\eps)^\vee$, which gives this maximizer just like in \eqref{eq:varprob1}, is thus a singleton. 
	That the subdifferential $\underline\partial F_\eps$ gives a singleton is
	naturally true in the case of Gâteaux differentiability as in Theorem
	\ref{thm:258cropped}.
	(See \cite{Babru-Precupanu} just above Proposition~2.33 and also note Proposition~2.40 in this context.) 
	Strict concavity of the
	energy functional also connects to strict monotonicity of its
	superdifferential and the Hohenberg--Kohn theorem, see Section~\ref{sec:Mono}.
\end{remark}

\begin{corollary}\label{cor:proximal}
	Let $X, X^*$ be reflexive and strictly convex. Any solution of the regularized problem $x \in \overline\partial E_\eps(x^*)$ is connected to a ``physical'' solution of the unregularized problem by the proximal mapping
	\[
	\prox_{\eps F}(x) \in \overline\partial E(x^*) \subset \tilde X.
	\]
\end{corollary}

\begin{proof}
	Take $x \in \overline\partial E_\eps(x^*)$, which is equivalent to $-x^* \in \underline\partial F_\eps(x)$ due to Lemma~\ref{lemma:equivdiff}. This subdifferential is even a singleton by Theorem~\ref{thm:258cropped} and it holds that
	\[
	-x^* = \eps^{-1} J(x - \prox_{\eps F}(x)).
	\]
	Since the duality map is bijective we apply $J^{-1}$ to get
	\[
	-\eps J^{-1}(x^*) = x - \prox_{\eps F}(x).
	\]
	We have here also used the fact that the duality map is always a homogeneous function. Now, comparing this with \eqref{eq:Ereg2} in Theorem~\ref{thm:Ereg}, we conclude from $x \in \overline\partial E_\eps(x^*)$ that $\prox_{\eps F}(x) \in \overline\partial E(x^*)$.
	It finally follows from the definition of $E$ in \eqref{eq:EdefF} and $F(x) = +\infty$ whenever $x \in X\setminus\tilde X$ that this solution is always in $\tilde X$.
\end{proof}

\section{Regularized Kohn--Sham iteration}

\subsection{Basic setup}

Assume that $x^*\in X^*$ has been given and we want to obtain the ground-state energy $E(x^*)$ from \eqref{eq:Edef} and a corresponding minimizer $x$ from \eqref{eq:argminenergy}, meaning that we must satisfy the equivalent conditions 
\begin{equation*}
-x^* \in \underline\partial F(x) \iff x \in\overline \partial E(x^*).
\end{equation*}
Transforming to the regularized energy functional with the help of Theorem~\ref{thm:Ereg}, we obtain
\begin{align*}
x \in \overline\partial E(x^*) &= \overline\partial E_\eps(x^*) + \eps J^{-1}(x^*),\\
E(x^*) &= E_\eps(x^*) + \eps\phi(x^*).
\end{align*}

Parallel to that we assume the existence of a reference functional $\tilde F^0 : \tilde X \rightarrow \mathbb{R}$ that belongs to the Kohn--Sham system and captures parts of the system's internal physics and leads to a variational problem that is supposedly easier to solve. Just like with $\tilde F$ we derive the regularized functionals $F_\eps^0$ and $E_\eps^0$ and set up the reference problem in analogous fashion
\begin{align*}
-\xks^* \in \underline\partial F_\eps^0(x) &\iff x \in \overline\partial E_\eps^0(\xks^*), \\
-x^* \in \underline\partial F_\eps(x) &\iff x \in \overline\partial E_\eps(x^*).
\end{align*}
Note that the minimizer state $x$ is the same in both cases, which makes it necessary to choose a different and at this stage undetermined potential for the reference system, the Kohn--Sham potential $\xks^*$. The $v$-representability problem that strikes the standard Kohn--Sham construction at this point 
(not being able to choose a potential such that the \emph{same} state is the solution) plays no role in the regularized version. By differentiability of $F_\eps$ and $F_\eps^0$, achieved through regularization, such a potential always exists as the (negative) derivative of the regularized functionals
\begin{align*}
\xks^* = - \nabla F_\eps^0 (x), \quad   x^* = - \nabla F_\eps (x).
\end{align*}
Subtraction of those two equations yields the first step in a self-consistent iteration scheme that begins with a density $x_1$ as a guess to the minimizer $x$ and eventually converges to $\xks^*$. One sensible initial guess used in Theorem~\ref{th:ks} later is $x_1 \in \overline\partial E_\eps^0(x^*)$, which means just taking the fixed external potential as a rough first approximation to the Kohn--Sham potential also capturing internal effects. The next iteration towards $\xks^*$ is then given by
\begin{equation}\label{eq:ks1}
x^*_{i+1} = x^* + \nabla F_\eps (x_i) - \nabla F_\eps^0 (x_i).
\end{equation}

The expression $\nabla F_\eps - \nabla F_\eps^0$, although in practice not known explicitly, is at least accessible through approximations since major contributions are expected to cancel due to similar physical effects in both systems. In DFT, this potential is known by the name ``Hartree exchange--correlation'' and subsumes all interaction effects that are lost in the non-interacting reference system. 
The second part of the Kohn--Sham iteration is then the solution of the (simple or simpler) reference system
\begin{align*}
x_{i+1} \in \overline\partial E_\eps^0(x^*_{i+1}).
\end{align*}
 
The stopping condition is $x^*_{i+1} = - \nabla F_\eps^0 (x_i)$ because this gives $x^* = -\nabla F_\eps (x_i)$ from \eqref{eq:ks1}, which means $x_i$ is the state yielding minimal (regularized) energy and consequently $x^*_{i+1} = \xks^*$. Note that this Kohn--Sham potential belongs to the regularized reference system and the resulting state is ``unphysical'' and in general $x \notin \tilde X$. 
%
%
%
Effectively we introduced two transformational layers to the problem,
firstly the reference system that gets connected to the sought-after
solution with the whole Kohn--Sham procedure, and secondly the
regularization. Invoking Theorem~\ref{thm:Ereg} and Corollary~\ref{cor:proximal}, it is possible to
translate both the ground-state energy and the corresponding state back to the original unregularized layer.

The important question of convergence of the Kohn--Sham scheme remains unaddressed up to this point. We shall see that, to guarantee at least weak convergence, the iteration must be slightly changed and not the full step from $x_i$ to $x_{i+1}$ is to be taken (following the ODA of the extended KS scheme in \cite{Cances2001}).

\subsection{Kohn--Sham iteration scheme}
\label{sec:ks}

We can now formulate the iteration scheme and prove a weak form of convergence in terms of the energy.

\begin{theorem}\label{th:ks}
	Let $X, X^*$ be reflexive and strictly convex, $E^0$ finite everywhere,
	$x^*\in X^*$ fixed, and set $x_1^*= x^*$ and select $x_1\in \overline\partial E_\eps^0(x^*)$. 
	Iterate $i=1,2,\dots$ according to:
	\begin{enumerate}[(a)]
		\item Set $x_{i+1}^* = x^* + \nabla F_\eps(x_i) - \nabla F_\eps^0(x_i)$ and stop if $x_{i+1}^* = -\nabla F_\eps^0(x_i) = \xks^*$.
		\item Select $ x_{i+1}' \in\overline\partial E_\eps^0(x_{i+1}^*)$.
		\item Choose $t_i\in (0,1]$ such that for
		$x_{i+1} = x_i + t_i( x_{i+1}' - x_i)$ one still has
		$$\langle \nabla F_\eps(x_{i+1}) +x^*,  x_{i+1}' - x_i \rangle \leq 0.$$
	\end{enumerate}
	Then the strictly descending sequence $\{ F_\eps(x_i) + \langle x^*,x_i\rangle \}_i$ converges as a sequence of real numbers to
	\begin{equation*}
	e_\eps(x^*) = \inf_i \{  F_\eps(x_i) + \langle x^*,x_i\rangle \} \geq E_\eps(x^*).
	\end{equation*}
	Thus, $e_\eps(x^*) + \eps\phi(x^*)$ is an upper bound for the ground-state energy $E(x^*)$.
\end{theorem}

\begin{proof}
First, we demonstrate that the superdifferential of $E^0_\eps$ is always non-empty such that step (b) can be performed (as well as the initial selection of $x_1$).
Note that this is equivalent to showing $\overline\partial E^0$ everywhere non-empty because of {\eqref{eq:Ereg1}} in Theorem~{\ref{thm:Ereg}}. Since $E^0$ is closed concave it is guaranteed to be weak-* upper semi-continuous, which is equivalent to weak semi-continuity because $X$ is reflexive. But weak (semi-)continuity always implies strong (semi-)continuity.
Finiteness of $E^0$ everywhere means $\dom(E^0) = X^*$ which yields also $\dom(\overline\partial E^0) = X^*$ by Corollary 2.38 in~\mbox{\cite{Babru-Precupanu}}. This gives us a non-empty superdifferential of $E^0_\eps$ everywhere.\\
\indent Next, we note that $F_\eps$ and $F_\eps^0$ are both G\^{a}teaux differentiable by Theorem~\ref{thm:258cropped} and start by checking the directional derivative of $F_\eps + x^*$ at $x_i$ in the step direction $x_{i+1}' - x_i$. From (a) in the KS scheme above, we have
	\begin{align*}
	&\langle \nabla F_\eps(x_i) + x^*,  x_{i+1}' - x_i \rangle \\
	&\quad = \langle x_{i+1}^* + \nabla F_\eps^0(x_i),  x_{i+1}' - x_i \rangle.
	\end{align*}
	If $x_{i+1}^* = - \nabla F_\eps^0(x_i)$ then from (a) $x^* = -\nabla
	F_\eps(x_i)$, which is the desired ground-state solution, and we have
	converged to the KS potential $\xks^*= x_{i+1}^*$. Otherwise, because of
	$x_{i+1}' \in\overline\partial E_\eps^0(x_{i+1}^*)$ and $x_i \in \overline\partial
	E_\eps^0(-\nabla F_\eps^0(x_i))$, we can invoke
	Lemma~\ref{lemma:monotone} for the strictly concave $E_\eps^0$ (see
	Remark~\ref{remark:strictly-concave}) and get
	\begin{align*}
	&\langle \nabla F_\eps(x_i) + x^*,  x_{i+1}' - x_i \rangle \nonumber\\
	&\quad = \langle x_{i+1}^* + \nabla F_\eps^0(x_i), x_{i+1}' - x_i
	\rangle <0.
	\end{align*}
	This means that we can always choose a maximal step size $t_i$ in (c) above such that
	\[
	F_\eps(x_{i+1}) + \langle x^*,x_{i+1}\rangle < F_\eps(x_i) + \langle x^*,x_i\rangle.	
	\]
	This sequence 
	$\{ F_\eps(x_i) + \langle x^*,x_i\rangle \}_i$ is by definition bounded below by $E_\eps(x^*)$ and hence convergent. \\
\indent Finally, we set  
	$e_\eps(x^*) = \lim_{i\to\infty} ( F_\eps(x_i) + \langle x^*,x_i\rangle )$. By \eqref{eq:Ereg1} of Theorem~\ref{thm:Ereg}
	\[
	F_\eps(x_i) + \langle x^*,x_i\rangle \geq E_\eps(x^*) = E(x^*) - \eps \phi(x^*).
	\]
	Consequently, $e_\eps(x^*) + \eps \phi(x^*)$ is an upper bound to the (non-regularized) ground-state energy $E(x^*)$.
\end{proof}

\begin{remark}\label{remark:gs}
The first part of the proof of Theorem~\ref{th:ks}, using finiteness  of $E^0$ to get a non-empty superdifferential $\overline\partial E^0(x^*)$ (or also $\overline\partial E^0_\eps(x^*)$) for all potentials $x^* \in X^*$, shows that we are apparently able to always find a state $x$ of minimal (non-interacting) energy. By Lemma~{\ref{lemma:equivdiff}} and Eq.~{\eqref{eq:varprob1}}
\begin{align*}
x \in \overline\partial E^0(x^*) &\iff -x^* \in \underline\partial F^0(x)\\
&\iff E^0(x^*) =  F^0(x) + \langle x^*,x \rangle
\end{align*}
we have that the infimum in \eqref{eq:Edef} with $\tilde F$ replaced by $\tilde F^0$ (or $F^0$) is a minimum. Reflexivity is needed to access the full range of superdifferentials that are defined in $X$ instead of $X^{**}$ (see discussion after Eq.~\eqref{eq:superdiff}), which is of no significance since reflexivity establishes $X \simeq X^{**}$.
\end{remark}
\begin{remark}\label{remark:wagner}
The modification of the Kohn--Sham iteration to include a reduced step size (see (c) in Theorem~\ref{th:ks}) 
and also the convergence of the energy quantity $F_\eps + x^*$ were modeled after Wagner et al.~\cite{Wagner2013}. These authors mistakenly claim that ``the KS algorithm described above is guaranteed to converge''. But neither was convergence proven in the usual sense within Banach spaces nor must a converging sequence of potentials lead to the correct $\xks^*$. Additionally, the work of Wagner et al.~\cite{Wagner2013} does not use the regularized version of the functionals and thus differentiability ($v$-representability) cannot be guaranteed. Further investigations of the Kohn--Sham iteration scheme within the framework established here is needed to determine whether a stronger version of convergence can be achieved.
\end{remark}

\section{Applications}
\label{sec:applications}

\subsection{Standard density-functional theory}
\label{sec:standard-dft}

The standard setting of DFT for an $N$-electron quantum system governed by Coulombic repulsion was pioneered by Lieb \cite{Lieb83}, but without the tools of Moreau--Yosida regularization or a study of the Kohn--Sham iteration scheme. Adopting the setting to our framework is straightforward. \\
\indent We now change to standard notation
and set $x=\rho$ and $x^*=v$. 
The Levy--Lieb functional $\tilde F$ is defined on the state space of physical densities~\cite{Lieb83}, the $N$-representable densities,
\begin{equation*}
\tilde X = \{ \rho \in L^1(\R^3) \mid \rho \geq 0,  \nabla \sqrt{\rho} \in L^2(\R^3), \|\rho\|_1 = N \},
\end{equation*}
which include all states of finite kinetic energy. Furthermore, $\tilde X \subset L^3(\R^3)\cap L^1(\R^3)$ by Sobolev embedding~\cite{Lieb83}. The functional is then given as the expectation value of the universal part of the standard Hamiltonian,
\begin{equation*}
H_{\lambda} = -\frac{1}{2} \sum_{i} \nabla_i^2 + \lambda \sum_{i<j} \frac{1}{r_{ij}},
\end{equation*}
with a constrained-search over all wave functions $\psi$ that yield a given density $\rho \in \tilde X$ (denoted $\psi\mapsto\rho$),
\begin{equation*}
\tilde F(\rho) = \inf_{\psi \mapsto \rho} \{ \langle \psi, H_1 \psi \rangle \}.
\end{equation*}
Note that $\lambda=1$ corresponds to the interacting system. The functional for the non-interacting ($\lambda=0$) Kohn--Sham system is similarly given by
\begin{equation*}
\tilde F^0(\rho) = \inf_{\psi \mapsto \rho} \{ \langle \psi, H_0 \psi \rangle \}
\end{equation*}
or a restriction of the minimization domain to only those $\psi$ that are Slater determinants. \\
\indent By the convex conjugate transformations \eqref{eq:Edef} and \eqref{eq:Fdef}, the functionals $E$ and $F$ are defined on the larger space of densities $L^3(\R^3) \cap L^1(\R^3)$ and the dual space of potentials $L^{3/2}(\R^3) + L^\infty(\R^3)$, which includes singular Coulomb potentials. With $X= L^3 \cap L^1$ and $X^* = L^{3/2} + L^\infty$, the convex Lieb functional $F$ is nowhere differentiable 
but the set of $v$-representable densities $\dom(\underline\partial F)$ is dense in the set of $N$-representable densities $\dom(F)$~\cite[Theorem 3.14]{Lieb83}.
The concave ground-state energy $E$ is finite on $X^*$~\cite[Theorem 3.1(iii)]{Lieb83}, 
superdifferentiable at all potentials $v$ that support an electronic ground state, and differentiable whenever the ground-state density is nondegenerate.
We have 
\begin{equation*}
-v \in \underline  \partial F(\rho) \iff \rho \in \overline \partial E(v).
\end{equation*}
Kvaal and Helgaker~\cite{KH2015} showed that  
\begin{equation*}
\rho \in \overline \partial E(v) \iff E(v) = \tr \Gamma_\rho H(v),
\end{equation*}
where $H(v)=H_1 + \sum_i v(r_i)$ is the Hamiltonian with potential $v$ and
$\Gamma_\rho$ is a ground-state density matrix with density $\rho$. It follows that 
$\overline \partial E(v)$ contains precisely all (ensemble) ground-state densities associated with the potential  $v$,
while $\underline \partial F(\rho)$ contains all potentials associated with the ground-state density $\rho$. Clearly, if $H(v)$
does not have a ground state, then $\overline \partial E(v) = \emptyset$. \newline
\indent The Hohenberg--Kohn theorem, as the cornerstone of DFT, drops out ``for free'' in the regularized version from strict monotonicity of $\overline \partial E_\eps$, see Section~\ref{sec:Mono}. This property follows in turn from $E_\eps$ being strictly concave. Yet without Moreau--Yosida regularization this strict concavity of $E$ is not at hand and arriving at the usual Hohenberg--Kohn theorem requires a refined analysis of the ground-state density that must not be zero on a set of nonzero measure~\cite{Lammert2018}. \newline
\indent To continue, the choice $X=L^3\cap L^1$ ($X^*=L^{3/2}+L^\infty$)
does not fit the framework developed here since the $L^1$-$L^\infty$ pair destroys reflexivity. A simple solution lies in just widening the density space to $X = L^3(\R^3)$, which includes $\tilde X$. The dual space for potentials is then restricted to $X^* = L^{3/2}(\R^3)$, which is reflexive as required by the above theorems. Coulomb potentials on all of $\R^3$ are then ruled out but are still included if the spatial domain is limited to a bounded $\Omega \subset \R^3$. Also, the non-interacting energy $E^0$ is everywhere finite on $X$ by the fact that $E$ has this property on $L^{3/2} + L^\infty$ (as remarked above). Thus, Theorem~\ref{th:ks} is applicable in this setting. Furthermore, under the assumption of reflexivity, and as noted in Remark~\ref{remark:gs}, finiteness of $E^0$ ($E$) gives a non-empty superdifferential $\overline\partial E^0(v)$ ($\overline\partial E(v)$). In DFT with $X=L^3$, if the corresponding density $\rho\in X$ is $N$-representable then we can even assign an associated ground-state wave function. This is similar to the domain (or box $[-l,l]^3$) truncation of \cite[Section III~A]{Kvaal2014}, where a ground state naturally exists for every potential. The setting of a free particle in infinite space, where clearly no ground state is at hand, must then be ruled out. The critical requirement is clearly reflexivity of $X$, which excludes the use of $L^1$ as a density space. \newline
%
%
%
%
\indent We remark that gauge symmetry can become complicated in a regularized setting. In Kvaal et al.~\cite[Section V~B]{Kvaal2014} it was shown that, in a Hilbert space setting where $v\in L^2(\Omega)$, the complications remain fairly mild. In the present example, with $v\in L^{3/2}(\Omega)$, the gauge symmetry becomes more unwieldy. We use the symmetry $E(v+c) = E(v) + cN$, with $c\in\mathbb{R}$ a constant shift, of the unregularized energy functional and the definition $E_{\epsilon}(v) = E(v) - \frac{\epsilon}{2} \|v\|_{3/2}^2$
to write
\begin{equation*}
  E_{\epsilon}(v+c) 
      = E_{\epsilon}(v) + cN - \frac{\epsilon}{2} \big( \|v+c\|_{3/2}^2 - \|v\|_{3/2}^2 \big).
\end{equation*}
Noting the functional derivative
\begin{equation*}
  \nabla \|v\|_{3/2}^2  = 2 \|v\|_{3/2}^{1/2} \, |v|^{1/2} \, \mathrm{sgn}(v)
\end{equation*}
and comparing to {Eq.~\eqref{eq:Ereg2}} above, we obtain
\begin{align*}
  \overline{\partial} &E_{\epsilon}(v+c) - \overline{\partial} E_{\epsilon}(v)  = \epsilon J^{-1}(v) - \epsilon J^{-1}(v+c) \\
    & = \epsilon \|v\|_{3/2}^{1/2} \, |v|^{1/2} \, \mathrm{sgn}(v) \\
&\quad - \epsilon \|v+c\|_{3/2}^{1/2} \, |v+c|^{1/2} \, \mathrm{sgn}(v+c),
\end{align*}
where sgn denotes the sign function. 
The fact that this difference does not vanish means that the potentials $v+c$ and $v$ map to different regularized ground-state densities.\newline
\indent Conventionally, the Kohn--Sham approach leads to the minimization problem (see the work of \cite{Hadjisavvas1984} and \cite{Laestadius2014b} for more technical details)
\begin{align}
E(v) &= \inf_{\psi} \{ F(\rho_\psi)+ \langle v,\rho_\psi\rangle   \} \nonumber \\
& = \inf_{\psi} \{ T(\psi) + J(\rho_\psi) + F_\mathrm{xc}(\rho_\psi)   +\langle v,\rho_\psi \rangle   \}, \label{eq:MinDet}
\end{align}
where $T(\psi)$ is the kinetic energy, $J(\rho_{\psi})$ is the Hartree term, and $F_{\mathrm{xc}}(\rho) = F(\rho) - F^0(\rho) - J(\rho)$ is the exchange-correlation functional.
Also, we have above used the notation $\rho_\psi$ to indicate that the density $\rho$ has been computed from $\psi$, 
i.e., 
\[
\rho_\psi(r) = \int_{\mathbb R^{3(N-1)}} |\psi|^2 \mathrm{d}r_2\cdots \mathrm{d}r_N.
\]
Glossing over the fact that $F(\rho)$ is not differentiable, the stationary condition for the above minimization gives the Kohn--Sham equations. When the iterative procedure defined in Theorem~\ref{th:ks} converges to the minimum, one obtains the Kohn--Sham potential $v_{\mathrm{KS}}$. Given this potential, it then only requires finding the ground state of a fixed, non-interacting Hamiltonian $H_0 + v_{\mathrm{KS}}$ in order to determine the Kohn--Sham wave function $\psi_{\mathrm{KS}}$ that solves the wave function optimization problem in \eqref{eq:MinDet} above.\\
\indent Another related problem is that of finding a Kohn--Sham potential $v_{\mathrm{KS}}$ that reproduces the density $\rho_{\mathrm{gs}}$ of the ground state of the interacting Hamiltonian $H_1+v_{\mathrm{ext}}$. This can be done using the optimization problem
\begin{equation}
\label{eq:LiebOpt} 
F^0(\rho_{\mathrm{gs}}) = \sup_{v} (E^0(v) - \langle v, \rho_{\mathrm{gs}} \rangle ).
\end{equation}
When the iterative procedure defined in Theorem~\ref{th:ks} converges to the minimum of the energy functional, it in fact solves this problem too {\it in a way that does not require a priori knowledge of $\rho_{\mathrm{gs}}$}. Instead, the interacting density $\rho_{\mathrm{gs}}$ is specified only implicitly by specifying the external potential $v_{\mathrm{ext}}$.\\
\indent In summary, although the Kohn--Sham wave function optimization problem \eqref{eq:MinDet} and the Lieb optimization problem \eqref{eq:LiebOpt} are distinct, the iterative procedure analyzed above addresses both problems. However,
%
%
	in solving the $v$-representability problem, we have introduced non-$N$-representable densities, that is, densities $\rho \notin \dom(F)$. The question then arises regarding the representation of such densities. The non-$N$-representable densities are obtained in step (b) of Theorem~\ref{th:ks} as a supergradient of the
	regularized non-interacting energy $ \rho_{i+1}' \in \overline \partial E_\eps^0(v_{i+1})$. Since 
	\begin{equation*}
	\overline \partial E_\eps^0(v_{i+1}) = \overline \partial E^0(v_{i+1}) - \eps J^{-1}(v_{i+1}),
	\end{equation*}
	we may first select an element of $\overline \partial E^0(v_{i+1})$ in the usual manner (by solving
	the non-interacting Schr\"odinger equation%
	) and then add the regularization correction $ - \eps J^{-1}(v_{i+1})$. See Section~VI.B in~\cite{Kvaal2014} for further details on how the regularization modifies the Kohn--Sham eigenvalue problem.

\subsection{Current-density-functional theories}

In current-density-functional theory (CDFT), both the paramagnetic current density and the total (physical) current density can be used together with the particle density $\rho$. We refer to \cite{Tellgren2012} and \cite{LaestadiusBenedicks2014} for a discussion of the choice of variables in CDFT. For the specific case of uniform magnetic fields, the current-density degrees of freedom can be reduced into a theory that has been named linear vector potential-DFT (LDFT) \cite{Tellgren2018}. \\
\indent The work of Lieb in DFT \cite{Lieb83} was in parts continued into paramagnetic CDFT in \cite{Laestadius2014}, where it was proven that each component of the paramagnetic current density is an element of $L^1(\mathbb R^3)$. However, since $L^1$ does not fulfill the requirements presented here, further analysis of function spaces for the paramagnetic formulation is needed. Nonetheless, we conjecture that each component of the paramagnetic current is an element of $L^1\cap L^q$ for $1<q<2$ and we suggest $L^{3/2}$ as a suitable space to choose. We moreover point out that the work in 
\cite{Laestadius2014b} only addressed the problem of $v$-representability by generalizing the work in \cite{Hadjisavvas1984} to include paramagnetic current densities. The problem of differentiability was not dealt with. The application of the theory outlined here to CDFT formulated with the paramagnetic current density is left for future work and will be based on the above conjecture.\\
%
%
%
%
%
%
\indent As far as the total (physical) current density is concerned, recent work has established a density-functional theory based on the
Maxwell--Schr\"odinger model~\cite{Tellgren2017}. In this theory, the potential space contains pairs $x^*=(v,B)$ of
electrostatic scalar potentials and magnetic fields. Holding $v$
fixed, the magnetic self energy plays the role of $\phi$ above,
yielding a ground state energy $E(x^*)$ that is already a Moreau--Yosida
regularization 
with respect to the argument $B$. The formalism
admits construction of a universal density functional $F(x)$ defined
for pairs $x = (\rho,\beta)$, where $\rho$ is the electron density and
$\beta$ a type of internal magnetic field that plays the role of an
independent variational parameter. Moreover, for any fixed $\rho$ in
its domain, the universal functional is differentiable with respect to
$\beta$. A further regularization with respect to $\rho$ results in
functionals $E_{\varepsilon}(x^*)$ and $F_{\varepsilon}(x)$ that are within
the scope of the above convergence result. \newline
%
%
%
%
\indent The Maxwell--Schr\"odinger model is itself an approximation to a more
complete description taking into account the quantized nature of the
light field that generates the internal magnetic field $\beta$ \cite{RuggiFlickAppel2018}. This
more complete description is based on the Pauli--Fierz Hamiltonian of
non-relativistic quantum electrodynamics \cite{RuggiFlickAppel2018,Spohn2004}, which describes the
interaction among charged particles (electrons and effective nuclei) by
the exchange of photons, the fundamental gauge bosons of the
electromagnetic force \cite{Ryder1996}. Consequently, the resulting
density-functional reformulation is a multi-component theory of fermions
and bosons \cite{Rugg2017} and we have two potentials  $x^{*}=(v,j)$ that act on the respective
particle families. Here $v$ is the usual electrostatic
scalar potential acting on the electrons and $j$ is an external
classical transversal charge current that acts on the photons. Using
the standard Maxwell relations, this external current can be directly
related to a unique classical magnetic field $B$ \cite{Rugg2017}. The conjugate pair
$x=(\rho,A)$ is then the usual electronic density and the transversal
electromagentic vector potential $A$ that is generated by the photons of
the coupled matter-photon system. Again, by using the Maxwell relations, 
the transversal vector potential $A$ is uniquely associated with an
internal magnetic field $\beta$. The above discussed Moreau--Yosida regularization can then be applied to quantum-electrodynamical DFT, and a rigorous
Kohn--Sham iteration scheme based on coupled Maxwell--Pauli--Kohn--Sham
equations can be introduced. We note that the presented
density-functionalization and generalization of the Kohn--Sham procedure are applicable in a straightforward manner to other coupled fermion-boson
problems. This highlights the applicability of the introduced approach beyond the usual confinement of traditional
electronic-structure theory.

\section*{Acknowledgments}
We are very grateful towards an anonymous referee who not only highlighted some important mistakes in a draft of this work but also hinted us towards possible solution schemes.
This work was supported by the Norwegian Research Council through the
CoE Hylleraas Centre for Quantum Molecular Sciences Grant No.~262695. 
A.L. is grateful for the hospitality received at the Max Planck Institute for the Structure and Dynamics of Matter in Hamburg, while visiting  MP and MR.
MP acknowledges support by the Erwin Schr\"odinger Fellowship J 4107-N27 of the FWF (Austrian Science Fund).
AL and SK was supported by ERC-STG-2014 under grant agreement No.~639508.
EIT was supported by the Norwegian Research Council through Grant No.~240674. 
TH is grateful to the Centre for Advanced Study at the Norwegian Academy of Science and Letters, Oslo, Norway, where parts of this work was carried out.

\end{document}